\documentclass[12pt,reqno,a4paper]{amsart}

\usepackage{amsmath,amsfonts,amsthm,amssymb,amsxtra,enumerate,mathtools,mathabx}
\usepackage[normalem]{ulem}
\usepackage[utf8]{inputenc}
\usepackage[T1]{fontenc}
\usepackage{lmodern}
\usepackage{bbm}
\usepackage{mathrsfs}
\usepackage{enumerate}

\newtheorem{theorem}{Theorem}
\newtheorem{proposition}[theorem]{Proposition}
\newtheorem{lemma}[theorem]{Lemma}

\theoremstyle{definition}

\theoremstyle{remark}

\newtheorem{remark}[theorem]{Remark}

\newcommand{\dd}{\, \mathrm{d}}
\newcommand{\eps}{\varepsilon}
\renewcommand{\epsilon}{\varepsilon}
\newcommand{\e}{\mathrm{e}}
\newcommand{\ii}{\mathrm{i}}
\newcommand{\id}{\mathbb{I}}

\newcommand{\N}{\mathbb{N}}
\newcommand{\norm}[2][]{{\left\|#2\right\|}_{#1}}   

\renewcommand{\phi}{\varphi}
\newcommand{\R}{\mathbb{R}}

\newcommand{\sclp}[2][]{{\left\langle#2\right\rangle} _{#1}}
\newcommand{\set}[2]{\left\{#1: \, #2\right\}}
\newcommand{\Sph}{\mathbb{S}}

\newcommand{\Z}{\mathbb{Z}}
\newcommand{\Hs}{\mathcal{H}}

\DeclareMathOperator{\supp}{supp}

\DeclareMathOperator{\Tr}{Tr}
\DeclareMathOperator{\tr}{Tr}

\begin{document}
	
	\title[Calogero type bounds in two dimensions
	]
	{Calogero type bounds in two dimensions}
	
	\author{Ari Laptev} 
	\address{Ari Laptev, Department of Mathematics, Imperial College London, London SW7 2AZ, UK and SPBU}
	\email{a.laptev@imperial.ac.uk}
	
	\author{Larry Read} 
	\address{Larry Read, Department of Mathematics, Imperial College London, London SW7 2AZ, UK}
	\email{l.read19@imperial.ac.uk}
	
	\author{Lukas Schimmer}
	\address{Lukas Schimmer, Institut Mittag--Leffler, The Royal Swedish Academy of Sciences, 182 60 Djursholm, Sweden}
	\email{lukas.schimmer@kva.se}
	\date{}
	
	
	\begin{abstract}
	For a Schr{\"o}dinger operator on the plane $\R^2$ with electric potential $V$ and Aharonov--Bohm magnetic field we obtain an upper bound on the number of its negative eigenvalues in terms of the $L^1(\R^2)$-norm of $V$. Similar to Calogero's bound in one dimension, the result is true under monotonicity assumptions on $V$. Our proof method relies on a generalisation of Calogero's bound to operator-valued potentials. We also establish a similar bound for the Schr{\"o}dinger operator (without magnetic field) on the half-plane when a Dirchlet boundary condition is imposed and on the whole plane when restricted to antisymmetric functions.
	\end{abstract}

	\maketitle
	
	\section{Introduction and main results}
	
	\noindent
	For a self-adjoint Schr{\"o}dinger operator $H_0-V=-\Delta-V$ on $L^2(\R^d)$ with non-negative potential $V\in L^{d/2}(\R^d)$ in dimension $d\ge 3$ the  celebrated Cwikel--Lieb--Rozenblum inequality \cite{Cw,L,R} provides an upper bound on the number $N(H_0-V)$ of negative eigenvalues. This so-called CLR bound states that	
	\begin{align}
	N(H_0-V)\le C_d \int_{\R^d} V(x)^{d/2}\dd x
	\label{eq:CLR}
	\end{align}
	with a constant $C_d$ independent of $V$. In dimensions $d=1$ and $d=2$ such a bound cannot hold true since $0$ is a resonance state of the spectrum, i.e.~for any potential $V\in\mathcal{C}_0^\infty(\R^2)$ with $\int_{\R^d}V(x)^{d/2}\dd x>0$ the Schr{\"o}dinger operator $H_0-V$ already has at least one negative eigenvalue. In dimension $d=2$, the operator $H_0-V$ may not even be semibounded under the assumption $V\in L^1(\R^2)$.  For \eqref{eq:CLR} to be true, repulsive terms are needed, which can be provided through additional electric potentials, magnetic fields, boundary conditions and/or particle statistics.
	
	In dimension $d=1$ one can for example restrict to $\R_+$ with Dirichlet condition at the origin. If the potential $V$ is furthermore non-increasing, i.e.~$V(x)\ge V(x')$ for $x\le x'$, the well known Calogero bound \cite{C}
	\begin{align}
	N(H_0-V)\le \frac{2}{\pi} \int_{\R_+} V(x)^{1/2}\dd x
	\label{eq:Calogero}
	\end{align}
	holds true. Our first result can be seen as a generalisation of this bound to dimension $d=2$. To this end we define the half-plane $\R^2_+ = \set{x=(x_1,x_2) \in\R^2}{ x_1>0, \, x_2\in \R}$.  

\begin{theorem}\label{th:half}
Let $V\in L^1(\R_+^2),V\ge0$  be non-increasing in $x_1$, i.e.~ $V(x_1,x_2)\ge V(x_1',x_2)$ for $x_1\le x_1'$.  Then the number of negative eigenvalues of the operator $H_0-V=-\Delta-V$ with Dirichlet boundary condition on $x_1=0$ satisfies the inequality 
		
		\begin{equation*}
			N(H_0-V) \le C \, \int_{\R^2_+} V(x)\, \dd x, 
		\end{equation*}
		where the constant $C\le 4.32$ is independent of $V$.
\end{theorem}

	In dimension $d=2$ the resonance state at zero can also be removed by adding an extra positive ``Hardy term'' $\frac{b}{x^2}$ with $b>0$ to the operator. However, for general $V\in L^1(\R^2)$, the operator $H_{b}-V:=-\Delta+\frac{b}{x^2}-V$ may still not be semibounded, requiring further restrictions on the potential. In \cite{L1} a CLR bound \eqref{eq:CLR} for $N(H_{b}-V)$ was proved for a class of potentials $V(x)=V(|x|)$ depending only on $|x|$. In  \cite{LN} this result was extended to a class of potentials belonging to $L^1(\R_+, r\dd r; L^p(\Sph^1))$, $1<p\le\infty$,  still assuming that there is an additional positive Hardy term. In polar coordinates the bound then takes the form
	\begin{align*}
	N(H_b-V)\le C_{b,p}\int_{\R_+}\norm[L^p(\Sph^1)]{V(r,\cdot)}r\dd r\,.
	\end{align*}
	In \cite{BEL} the authors noticed that the Hardy term could be provided by the Laplacian with an Aharonov--Bohm magnetic vector potential $A$ with non-integer flux $\Psi$, see \cite{LW2}. 
	For the operator $H_A -V=(-\ii\nabla+A)^2-V$ with  potential $V\in L^1(\R_+, r\dd r; L^\infty(\Sph^1))$ the authors proved that 
	\begin{align*}
	N(H_A -V)\le C_\Psi\int_{\R_+}\norm[L^\infty(\Sph^1)]{V(r,\cdot)}r\dd r\,.
	\end{align*}
	Only in the case $V(x)=V(|x|)$ (where the sharp value of the constant $C_\Psi$ was obtained in  \cite{L2}) the right-hand side coincides with the $L^1(\R^2)$-norm of $V$ as in the CLR bound~\eqref{eq:CLR}. 
	
	In our paper we will present a bound \eqref{eq:CLR} for potentials $V\in L^1(\R^2)$ satisfying a monotonicity condition. To be more precise, we consider the Schr{\"o}dinger operator with Aharonov--Bohm type magnetic field, i.e.~the operator
	\begin{align*}
		H_A=(-\ii\nabla +A(x))^2
	\end{align*}
	on $L^2(\R^2)$ with $A:\R^2\to\R^2$ defined, in polar coordinates, as
	\begin{align*}
		A(r,\varphi)=\frac{\psi(\varphi)}{r}(\sin\varphi,-\cos\varphi)
	\end{align*}
	where $\psi\in L^1(\Sph^1)$. The flux is given by $\Psi=(2\pi)^{-1}\int_{\Sph^1} \psi(\varphi)\dd\varphi$. Our main result is the following.

	\begin{theorem}\label{th:AB}
		Let $V\in L^1(\R^2), V\ge0$ be a potential that is non-increasing along any ray from the origin, i.e.~in polar coordinates $V(r,\phi) \ge V(r',\phi)$ for $r\le r'$. If $\Psi\notin\Z$ then the number of negative eigenvalues of the operator $H_A-V$ satisfies the inequality
		\begin{align*}
			N(H_A-V)\le  C_\Psi\int_{\R_+}\norm[L^1(\Sph^1)]{V(r,\cdot)}r\dd r =C_\Psi\int_{\R^2}V(x)\dd x
		\end{align*}
		where the constant $C_\Psi$ is independent of $V$. 
	\end{theorem}

\begin{remark}
The results of Theorem \ref{th:half} and Theorem \ref{th:AB} cannot hold for general integrable potential $V$ even under the additional assumption that $V$ is smooth and has compact support. To see this, one can consider any potential $V\in\mathcal{C}_0^\infty(\R^2)$ with $\int_{\R^2}V(x)\dd x>0$. Since $0$ is a resonance state of the Schr{\"o}dinger operator on $\R^2$ there exists $u\in\mathcal{C}_0^\infty(\R^2)$ with $\sclp[2]{u,(-\Delta-V)u}<0$. Simultaneously translating $V$ and $u$ to $\widetilde{V}$ and $\widetilde{u}$ such that both are entirely supported in the half-plane $\R_+^2$ then shows that the operator in Theorem \ref{th:half} with potential $\widetilde{V}$ also has at least one negative eigenvalue. For the magnetic operator, we can argue as in \cite[Section 5]{BEL} and note that, since $\mathrm{curl} A=0$ on the simply connected set $\R_+^2$, the Poincar{\'e} transformation yields a gauge transform such that $H_A$ and $-\Delta$ are equivalent in $\R_+^2$. This shows that the operator in Theorem \ref{th:AB} with potential $\widetilde{V}$ has at least one negative eigenvalue. 
\end{remark}

Finally, we note that a positive Hardy term can also be provided by restricting the Laplacian to antisymmetric functions (see for example \cite{HL}), i.e.~by considering $-\Delta-V$ on the Hilbert space $L^2_{\mathrm{as}}(\R^2)=\set{u\in L^2(\R^2)}{u(x_1,x_2)=-u(x_2,x_1)}$.
\begin{theorem}\label{th:as}
Let $V\in L^1(\R^2), V\ge0$ be a potential that is non-increasing along any ray from the origin, i.e.~in polar coordinates $V(r,\phi) \ge V(r',\phi)$ for $r\le r'$.  Then the number of negative eigenvalues of the operator $H_{\mathrm{as}}-V=-\Delta-V$ on the space $L^2_{\mathrm{as}}(\R^2)$ of antisymmetric functions satisfies 
\begin{align*}
	N(H_{\mathrm{as}}-V) \le C \, \int_{\R^2} V(x)\, \dd x, 
\end{align*}
where the constant $C\le 5.43$ is independent of $V$.    
\end{theorem}
Substantial efforts  were made (so far without any success) in finding necessary and sufficient conditions on a class of potentials that provide a finite number of negative eigenvalues for a two-dimensional Schr{\"o}dinger operator and also necessary and sufficient conditions for the validity of the Weyl asymptotics. In \cite{BL} the authors gave examples of different non-Weyl law formulae under the condition $V\in L^1(\R^2)$. Some upper estimates for $N(H_0-V)$ in the two-dimensional case were obtained in papers \cite{GN, KS, LS1, LS2, MV, MW, St}. In \cite{Sh} the author gives estimates for the number of negative eigenvalues of a two-dimensional Schr{\"o}dinger operator in terms of $L \log L$ type Orlicz norms of the potential and proves a conjecture by N.~N.~Khuri, A.~Martin and T.~T.~Wu \cite{KMW} (see also \cite{CKMW}). The fact that the classes $L \log L$ of potential functions are relevant to estimates for $N(H_0-V)$ was first discovered by M.~Solomyak \cite{Sol}.

The proofs of Theorem \ref{th:half}, Theorem \ref{th:AB} and Theorem \ref{th:as} all use the same central idea and employ the ``lifting argument'' presented in \cite{LW1} (see also \cite{HLW} and \cite{FLW}). In each setting a Hardy inequality holds, which allows us to add a Hardy term at the expense of reducing the `kinetic' part of the operator. Using the structure of the operators involved, all three problems can then be reduced to studying corresponding one-dimensional differential operators with operator-valued potentials. The Hardy term allows us to prove Calogero type bounds for these operators in terms of $1/2$ Schatten norms of the operator-valued potentials. Lastly, upper bounds on these norms can be established through Lieb--Thirring type inequalities \cite{LT}. To this end, in the case of Theorem \ref{th:AB}, we will prove a new inequality of this type for the magnetic operator on $L^2(\Sph^1)$, following the proof in \cite{W} of a similar bound on $L^2(\R)$. Note that for such operators in \cite{DELL} the authors obtained a range of Keller--Lieb--Thirring inequalities for the lowest eigenvalue. In our proofs we will also show that the operators involved can be defined as semibounded, self-adjoint operators.  

As a first step we will consider one-dimensional differential operators with operator-valued potentials in Section \ref{sec:CalOp}. Under the assumed presence of an additional Hardy term, we will establish Calogero-type bounds for these operators.  As a corollary, we will obtain a generalisation of  Calogero's bound \eqref{eq:Calogero} to operator-valued potentials. Sections \ref{sec:half}, \ref{sec:AB} and \ref{sec:as} then contain the proofs of our theorems.
	
	
\section{Calogero type bounds for 1D operators with operator-valued potentials and Hardy term}\label{sec:CalOp}	
Let $\mathcal H$ be a separable Hilbert space with inner product $\sclp[\Hs]{\cdot,\cdot}$ and let $\mathcal V(t)$, $t\ge0$, be an operator-valued function (potential), whose values are compact, non-negative self-adjoint operators in 
$\mathcal H$.
By $\mathfrak{S}_p$, $p>0$,  we denote the Schatten class of compact operators whose s-numbers satisfy the inequality 
\begin{align*}
\sum_{n=1}^\infty  s_n^p <\infty.
\end{align*}
As usual $\mathfrak{S}_\infty$ denotes the class of compact operators. 
For $f\in\mathcal{C}^\infty(\overline{\R_+})$ with $f(t)>0,r\in\R_+$ we consider the Hilbert space $L^2(\R_+,f(t)\dd t;\Hs)$. For our applications it will be sufficient to consider $f(t)=1$ and $f(t)=t$, corresponding to the radial part of the Laplacian in one and two dimensions, respectively. We will prove the results in slightly more generality though, assuming that $f$ is non-decreasing and a polynomial. We expect the results to also hold if the latter assumption is weakened, for example to that of \cite[Theorem 15.2]{We} along with the assumption that $\sup_{k\in\Z}f(a^{k+1})/f(a^k)<\infty$ for some $a>1$.  The subspace $\mathcal{C}^\infty_0(\R_+;\Hs)$, i.e.~the set of  functions $U:\R_+\to\Hs$ which are smooth (with respect to the topologies on $\R_+$ and $\Hs$) with compact support in $\R_+$, is dense in $L^2(\R_+, f(t)\dd t;\Hs)$. If $\mathcal{V}\in L^1(\R_+;\mathfrak{S}_\infty)$ then, for any $b\ge0$, the quadratic form
\begin{align*}
\int_{\R_+}\left(\norm[\Hs]{\frac{\mathrm{d}}{\mathrm{d}t}U(t)}^2+\frac{b}{t^2}\norm[\Hs]{U(t)}^2-\sclp[\Hs]{U(t),\mathcal{V}(t)U(t)}\right)f(t)\dd t
\end{align*}
is well-defined on $\mathcal{C}^\infty_0(\R_+;\Hs)$. It corresponds to the symmetric operator
\begin{align*}
H_b-\mathcal{V}=-\frac{1}{f(t)}\frac{\mathrm{d}}{\mathrm{d}t} f(t)\frac{\mathrm{d}}{\mathrm{d}t}\otimes\id +\frac{b}{t^2}\otimes\id-\mathcal{V}\,.
\end{align*}
If $\mathcal{V}\in L^1(\R_+;\mathfrak{S}_\infty)$, then the quadratic form is semibounded. We can thus consider the self-adjoint Friedrichs extension of $H_b-\mathcal{V}$, which we refer to as the operator with Dirichlet boundary condition and which we continue to denote by $H_b-\mathcal{V}$. 
Finally, we say that the family of operators $\mathcal{V}(t)$ is non-increasing if $\mathcal V(t)\ge \mathcal V(t')$ for $0\le t\le t'$  in the usual sense of quadratic forms. Namely, for any fixed vector $U\in \mathcal H$
$$
\sclp[\Hs]{U,\mathcal{V}(t) U}  \ge \sclp[\Hs]{U,\mathcal{V}(t')U} , \qquad 0\le t\le t'\,.
$$
\begin{proposition}\label{prop:Cal}
	Let $\Hs$ and $f$ be as above. Furthermore let $\mathcal{V}\in L^1(\R_+;\mathfrak{S}_{1/2}),\mathcal{V}\ge0$ be non-increasing. Then the number of negative eigenvalues of the operator $H_b-\mathcal{V}$ with Dirichlet boundary condition satisfies
	\begin{align*}
	N(H_b-{\mathcal{V}})\le C_{b,f}\int_{\R_+}\sqrt{\norm[\mathfrak{S}_{1/2}]{\mathcal{V}(t)}} \dd t
	\end{align*}
	where the constant $C_{b,f}$ is independent of $\mathcal{V}$. 
	\end{proposition}
\begin{proof}
	For $U\in\mathcal{C}_0^\infty(\R_+;\Hs)$ there exists $k\in\N$ such that $\supp(U)\subset[2^{-k},2^{k}]$ and thus 
	\begin{align*}
		&\sclp[2]{U,(H_b-\mathcal{V})U}\\
		&=\sum_{k\in\Z}\int_{2^k}^{2^{k+1}}
		\left(\norm[\Hs]{U'(t)}^2+\frac{b}{t^{2}}\norm[\Hs]{U(t)}^2-\sclp[\Hs]{U(t),\mathcal{V}(t)U(t)}\right) f(t)\dd t\\
		&\ge \sum_{k\in\Z}\int_{2^k}^{2^{k+1}}\left(\norm[\Hs]{U'(t)}^2+\frac{b}{t^{2}}\norm[\Hs]{U(t)}^2-\frac{f(2^{k+1})}{f(2^k)}\sclp[\Hs]{U(t),\mathcal{V}(2^k)U(t)}\right) f(2^k)\dd t\,.
	\end{align*}
	Let $B_k:=\set{t\in\R_+}{t\in (2^k,2^{k+1})}$. Each of the quadratic forms above is closed on $H^1(B_k;\Hs)$. We can thus consider the corresponding self-adjoint operator $h_k$. By the variational principle $H_b-\mathcal{V}\ge\bigoplus_{k\in\Z}h_k$.

	Now look on each $h_k$ and see
	\begin{align*}
		&N(h_k)\\
		&=\sup\dim \set{F\subset H^1(B_k)}{ \int_{2^k}^{2^{k+1}}\left(\norm[\Hs]{U'}^2+\frac{b}{r^{2}}\norm[\Hs]{U}^2-\frac{f(2^{k+1})}{f(2^k)}\sclp[\Hs]{U,\mathcal{V}(2^k)U}\right)\dd t<0, U\in F}\\
		&=\sup\dim \set{F\subset H^1(B_0)}{ \int_{1}^{2}\left(\norm[\Hs]{U'}^2+\frac{b}{s^2}\norm[\Hs]{U}^2-2^{2k}\frac{f(2^{k+1})}{f(2^k)}\sclp[\Hs]{U,\mathcal{V}(2^k)U}\right) \dd s<0, U\in F}\\
		&\le\sup\dim \set{F\subset H^1(B_0)}{\int_{1}^{2}\left(\norm[\Hs]{U'}^2+\frac{b}{4}\norm[\Hs]{U}^2-2^{2k}\frac{ f(2^{k+1})}{f(2^k)}\sclp[\Hs]{U,\mathcal{V}(2^k)U}\right)\dd s<0, U\in F}.	\end{align*}
On $L^2((1,2);\Hs)$ we now consider the operator with constant coefficients 
\begin{align*}
-\frac{\mathrm{d}^2}{\mathrm{d}s^2}\otimes\id + \frac{b}{4}\otimes\id - 2^{2k}\frac{f(2^{k+1})}{f(2^k)}\mathcal{V}(2^k)
\end{align*}
and Neumann boundary conditions in the eigenbasis of the compact, non-negative operator $\mathcal V(2^k)$. Denoting by $\mu_n(2^k)\ge 0$ the eigenvalues of $\mathcal V(2^k)$ we find
\begin{align*}
		N(h_k)\le \sum_{n}\left(\# \set{m\in \mathbb{N}_0}{\pi^2 m^2+\frac{b}{4}< 2^{2k}\frac{f(2^{k+1})}{f(2^k)}\mu_{n}(2^k)}\right).
\end{align*}
Each of the summands above can be bounded by $R(b)2^k\sqrt{f(2^{k+1})/f(2^k)}\sqrt{\mu_n(2^k)}$ with some constant $R(b)$ depending only on $b$ (see Remark \ref{Re:Constant} below for an explicit choice) and thus
\begin{align*}
		N(h_k)\leq R(b)2^k\sqrt{\frac{f(2^{k+1})}{f(2^k)}} \sum_n \sqrt{\mu_n(2^k)}
		\leq R(b)2^k\sqrt{\frac{f(2^{k+1})}{f(2^k)}}\sqrt{\norm[\mathfrak{S}_{1/2}]{\mathcal{V}(2^k)}}.
\end{align*}
Importantly $\mathcal{V}(t)$ is a non-increasing operator-valued function in $t$ and thus
	\begin{align*}
	N(H_b-\mathcal{V})\le\sum_{k\in\Z}N(h_k)&\le R(b)\sum_{k\in\Z}2^k\sqrt{\frac{f(2^{k+1})}{f(2^k)}}\sqrt{\norm[\mathfrak{S}_{1/2}]{\mathcal{V}(2^k)}}\\
	&\le 2R(b)\sup_{k\in\Z}\sqrt{\frac{f(2^{k+1})}{f(2^k)}}\int_{\R_+}\sqrt{\norm[\mathfrak{S}_{1/2}]{\mathcal{V}(t)}}\dd t\,.
	\end{align*}
\end{proof}
	\begin{remark}[Explicit constant]\label{Re:Constant}
		Here we derive an explicit upper bound on $R(b)$: Consider the inequality $(\alpha m^2+\beta)/(m+1)^2\ge \alpha\beta/(\alpha+\beta)$, for $\alpha,\beta,m\geq 0$. Then for sets of the type above, we find  
		\begin{align*}
			\#\{ m\in \mathbb{N}_0: \, \alpha m^2  + \beta < \gamma^2 \}&\le \#\{ m\in \mathbb{N}_0: \, (\alpha m^2  + \beta)(m+1)^2/(m+1)^2  < \gamma^2 \}\\
			&\le \#\{ m\in \mathbb{N}: \, (\alpha\beta/(\alpha+\beta))m^2< \gamma^2 \}\le \gamma\sqrt{\frac{\alpha+\beta}{\alpha\beta}} . 
		\end{align*}
		A direct consequence of this is the upper bound  $R(b)\le\sqrt{\frac{b+4 \pi ^2}{\pi^2 b}}$.
		We note that we can consider a more general scheme, where we split $\mathbb{R}_+$ by the intervals $(a^k,a^{k+1})$, with $a>1$. Applying the above to this general setting, we see that 
		 \begin{align*}
		     C_{b,f}\le aR\!\left(\frac{4 (a-1)^2 b}{a^2}\right)\sup_{k\in\Z}\sqrt{\frac{f(a^{k+1})}{f(a^k)}}
		 \end{align*}
		 where $C_{b,f}$ is the constant in Proposition 5. 
	\end{remark}
An immediate consequence is a generalisation of Calogero's bound \eqref{eq:Calogero} to operator-valued potentials.
\begin{theorem}\label{th:OpCal}
	Let $\mathcal{V}\in L^1(\R_+;\mathfrak{S}_{1/2}),\mathcal{V}\ge0$ be non-increasing. Then the number of negative eigenvalues of the  operator $-\frac{\mathrm{d}^2}{\mathrm{d}t^2}\otimes\id-\mathcal{V}$ with Dirichlet boundary condition satisfies
\begin{align*}
	N\!\left(-\frac{\mathrm{d}^2}{\mathrm{d}t^2}\otimes\id-{\mathcal V}\right) \le C\, \int_{\R_+} \sqrt{\norm[\mathfrak{S}_{1/2}]{\mathcal V(t)}}\dd t,
\end{align*}
	where $C\le 8.63$ is a constant independent of $\mathcal V$.
\end{theorem}
\begin{proof} 
If $U\in\mathcal{C}_0^\infty(\R_+;\Hs)$, we can apply the standard Hardy inequality on $\R_+$ and obtain 
	\begin{align*}
		\int_{\R_+}
		\norm[\Hs]{U'(t)}^2 \dd t
		\ge \frac{1}{4} \, \int_{\R_+} \frac{\norm[\Hs]{U(t)}^2}{t^2}\dd t.
	\end{align*}
	
	Let $0< \vartheta<1$. Splitting the operator of the second derivative $-\mathrm{d}^2/\mathrm{d}t^2= -(1-\vartheta) \mathrm{d}^2/\mathrm{d}t^2 - \vartheta \mathrm{d}^2/\mathrm{d}t^2$ and using the Hardy inequality we have (with $f(t)=1$)
\begin{align*}
N(H_0-\mathcal{V}) \le N((1-\vartheta)H_{\vartheta/(1-\vartheta)}-\mathcal{V})\,,
\end{align*}
where we used that $\mathcal{C}_0^\infty(\R_+;\Hs)$ is a form-core of the operators involved.
We can now apply the proposition to obtain the desired result with $C=2R(\vartheta/4(1-\vartheta))(1-\vartheta)^{-1/2}$.

\end{proof}

\begin{remark}[Explicit constant]\label{Re:constant2}
    We find the upper bound $C\le 8.63$ by following Remark \ref{Re:Constant}. Dividing into intervals of order $a$ and optimising over $a$ and $\vartheta$ yields the best value when $a=1.92882$ and $\vartheta=0.928815$. The best $\vartheta$ being close to $1$ is, we believe, a consequence of our proof method despite even small $\vartheta$ being sufficient to remove the zero-modes in the Neumann bracketing argument.
\end{remark}

\noindent {\it Open problem.} We do not believe the constant $8.63$ in Theorem \ref{th:OpCal} to be optimal  and thus it would be interesting to find the sharp constant. If $\mathcal{V}$ is diagonal, then the bound holds with the same sharp constant $2/\pi$ as in the scalar case.

\section{Proof of Theorem \ref{th:half}}\label{sec:half}

To prove Theorem 1 we use the ``lifting argument'' developed in \cite{LW1}. Noting that $L^2(\R_+^2)$ is isomorphic to $L^2(\R_+;L^2(\R))$ and, using the structure of the Laplacian, we write 
\begin{align*}
H_0 - V =  -\frac{\mathrm{d}^2}{\mathrm{d}x_1^2} \otimes \id - W(x_1),
\end{align*}
where 
\begin{align*}
W(x_1) = \frac{\mathrm{d}^2}{\mathrm{d}x_2^2} + V(x_1,x_2) 
\end{align*}
is an operator-valued potential in $L^2(\R)$. By Fubini's theorem $V(x_1,\cdot)\in L^1(\R)$ for almost all $x_1\in\R_+$ and thus $W(x_1)$ is a self-adjoint operator on the domain $H^1(\R)$ for almost all $x_1\in\R_+$. As we will see below, its positive part  $W_+(x_1)$ is a compact operator. Furthermore on $\mathcal{C}_0^\infty(\R_+; H^1(\R))\subset L^2(\R_+;L^2(\R))$ the operator inequality
\begin{align*}
     -\frac{\mathrm{d}^2}{\mathrm{d}x_1^2} \otimes \id - W(x_1)
     \ge -\frac{\mathrm{d}^2}{\mathrm{d}x_1^2} \otimes \id - W_+(x_1)
\end{align*}
holds. As discussed in Section \ref{sec:CalOp} the operator on the right-hand side is semibounded and consequently the same holds true for the former operator on the left-hand side. We can thus consider their respective self-adjoint Friedrichs extensions, which we continue to denote by the same symbols and which still satisfy the operator inequality. This establishes that the operator in Theorem \ref{th:half} is well-defined.
By the variational principle 
\begin{align*}
N(H_0-V) \le N\left(-\frac{\mathrm{d}^2}{\mathrm{d}x_1^2} \otimes \id - W_+(x_1)\right).
\end{align*}
Note that $W_+(x_1)$  is a non-increasing operator-valued function. 
Applying Theorem \ref{th:OpCal} we find 
\begin{align}
N(H_0-V) \le C \int_{\R_+} \sqrt{\norm[\mathfrak{S}_{1/2}]{W_+(x_1)}} \dd x_1.
\label{OperCal}
\end{align}
For each $x_1$ we use the sharp result on the $1/2$ moments for a one-dimensional Schr{\"o}dinger operator \cite{HLT} (see also \cite{HLW}) saying that, for $v\in L^1(\R), v\ge0$, 
\begin{align}\label{hlt}
\Tr\left( -\frac{\mathrm{d}^2}{\mathrm{d}t^2} - v(t)\right)_-^{1/2} \le \frac12\, \int_{\R} v(t)\dd t.
\end{align}
Applying \eqref{hlt} to the operator $W_+(x_1)$ we have 
\begin{align*}
\sqrt{\norm[\mathfrak{S}_{1/2}]{W_+(x_1)}} \le \frac12\, \int_{\R} V(x_1,x_2) \dd x_2.
\end{align*}
Combining the latter inequality with \eqref{OperCal} completes the proof of Theorem 1.
\begin{remark}[Explicit constant]
    From Remark \ref{Re:constant2}, the best constant known to us is $C\approx 4.31244$.
\end{remark}	

\section{Proof of theorem \ref{th:AB}}\label{sec:AB}
Since the operator with flux $\psi$ is gauge equivalent to the operator with constant flux $\Psi$, we assume without loss of generality that $\psi=\Psi$. In polar coordinates the quadratic form of $H_A$ on $\mathcal{C}_0^\infty(\R^2\setminus\{0\})$  is given by
	\begin{align*}
		\sclp[2]{u,H_A u}=\int_{0}^\infty\int_{\Sph^1}\left(\left|\frac{\partial}{\partial r}u(r,\varphi)\right|^2
		+\frac{1}{r^2}\left|\ii\frac{\partial}{\partial \varphi}u(r,\varphi)+\Psi u(r,\varphi)\right|^2\right)\dd\varphi\, r\dd r
	\end{align*}
and on this space the Hardy inequality \cite{LW2} 
	\begin{align}
		\sclp[2]{u, H_A u}
		\ge c_\Psi \int_{\R_+}\int_{\Sph^1}\frac{|u(r,\varphi)|^2}{r^2}\dd\varphi\, r\dd r
		\label{HardyAB}
	\end{align}
	holds with $c_\Psi=\min_{k\in\Z}|\Psi+k|^2$. 
	
Let $0<\vartheta<1$. Splitting the operator $H_A= (1-\vartheta) H_A + \vartheta H_A$ and using \eqref{HardyAB} we have the operator inequality
	\begin{align*}
		H_A-V \ge \widetilde{H}_A -V,
	\end{align*}
	on $\mathcal{C}_0^\infty(\R^2\setminus\{0\})$ with 
	\begin{align*}
		\widetilde{H}_{A}  = (1-\vartheta) H_A + \vartheta\frac{c_\Psi}{|x|^2}\,.
	\end{align*}
	The space $L^2(\R_+\times\Sph^1,r\dd r\dd\varphi)$ is isomorphic to  $L^2(\R_+, r\dd r; L^2(\Sph^{1}))$ which allows us to write
	\begin{align}\label{eq:thm2}
	    \widetilde{H}_A -V
	    =-(1-\vartheta)\frac{1}{r}\frac{\mathrm{d}}{\mathrm{d}r}r\frac{\mathrm{d}}{\mathrm{d}r}\otimes\id+\vartheta\frac{c_\Psi}{r^2}\otimes\id- W(r)
	\end{align}
	with the operator-valued potential 
	\begin{align*}
	W(r)=-\frac{1-\vartheta}{r^2}\left(\ii\frac{\partial}{\partial \varphi}+\Psi\right)^2+V(r,\varphi)\,.
	\end{align*} 
	For almost all $r\in\R_+$, the potential is self-adjoint on $H^1(\Sph^1)$ by the assumptions on $V$ and we can consider its positive part $W_+(r)$, which is compact as we will show below.  Arguing similarly as in Section \ref{sec:half} the operators above are all semibounded and we can thus consider their respective Friedrichs extensions. This establishes that the operator $H_A-V$ in Theorem \ref{th:AB} is well-defined. Note that $W_+$ is a non-increasing operator-valued function in $r$ and thus by Proposition \ref{prop:Cal} and the variational principle  (note that $f(t)=t$ and $\sup_{k\in\Z}f(2^{k+1})/f(2^k)=2$)
	\begin{align}
		N(H_A-V)\le 2^{3/2}R\Big(\frac{ c_\Psi\vartheta}{1-\vartheta}\Big)(1-\vartheta)^{-1/2}\int_{\R^+}\sqrt{\norm[\mathfrak{S}_{1/2}]{W_+(r)}}\dd r\,.
		\label{eq:NPsi}
	\end{align}
	To bound the integrand we need a Lieb--Thirring type inequality for the $1/2$ moments of a  Schr{\"o}dinger operator $W(r)$ with constant magnetic vector potential on $L^2(\Sph^1)$, which we will prove in Lemma \ref{lem:LTa} below. Combining this bound with \eqref{eq:NPsi} we obtain the desired inequality
	\begin{align*}
		N(H_A-V)\le2^{3/2}R\Big(\frac{c_\Psi\vartheta}{1-\vartheta}\Big)(1-\vartheta)^{-1}d_\Psi\int_{\R^+}\int_{\Sph^1}V(r,\varphi)r\dd\varphi\dd r\,.
	\end{align*}
	\begin{lemma}\label{lem:LTa}
		Let $v\in L^1(-\pi,\pi), v\ge0$ and $\Psi\notin\Z$. Then the operator
		\begin{align*}
			h-v:=\left(\ii\frac{\mathrm{d}}{\mathrm{d}\varphi}+\Psi\right)^2-v
		\end{align*}
		with periodic boundary conditions satisfies
		\begin{align*}
			\tr\left(\left(\ii\frac{\mathrm{d}}{\mathrm{d}\varphi}+\Psi\right)^2-v\right)_-^{1/2}\le d_\Psi\int_{-\pi}^\pi v(\varphi)\dd\varphi
		\end{align*}
		with a constant $d_\Psi$ independent of $v$.
	\end{lemma}
	We will see that $d_\Psi\to\infty$ as $\Psi\to k\in\Z$.
	Our proof will follow \cite{W}, where a Lieb--Thirring inequality for the $1/2$ moments of a Schr{\"o}dinger operator on $\R$ (i.e.~the bound \eqref{hlt} without the sharp constant) was proved. The main idea of the proof in \cite{W} is to use Neumann bracketing whereby $\R$ is partitioned into disjoint intervals that each support at most one eigenvalue. Importantly the  length of each interval (compared to the $L^1$-norm of the potential on the interval) can be chosen to be uniformly bounded from below. To achieve the latter on the bounded set $(-\pi,\pi)$ we will use that for $\Psi\notin\Z$ the operator $h$  does not have any zero-modes and that $h-v$ does not admit any negative eigenvalues if $\int_{-\pi}^\pi v(\varphi)\dd \varphi$ is small. The desired partition can then be constructed with multiplicity 2. The arguments of \cite{W} all carry over, as we will show in detail below.  We start with recalling the following from \cite[Lemmata 1 and 2]{W}.
	\begin{lemma}\label{lem:W}
		Let $q\in L^1(0,\ell), q\ge0$ and consider the operator $-\frac{\mathrm{d}^2}{\mathrm{d} \varphi^2}-q$ on $L^2(0,\ell)$ with Neumann boundary conditions. 
		\begin{enumerate}[(i)]
			\item With $g$ denoting the inverse of the strictly increasing function $x\tanh x$ on $\R_+$, the lowest eigenvalue $\lambda_1$ of the operator is bounded as
			\begin{align*}
				|\lambda_1|^{1/2}\le \frac{g(\ell \int_0^\ell q(\varphi)\dd \varphi)}{\ell}\,.
			\end{align*}
			\item If 
			\begin{align*}
				\ell\int_0^\ell q(\varphi)\dd\varphi\le 3
			\end{align*}
			then the operator has a single negative eigenvalue. 
		\end{enumerate}
	\end{lemma}
	We now give the proof of Lemma \ref{lem:LTa}.
	\begin{proof}[Proof of Lemma \ref{lem:LTa}]
		For $\lambda<\min_{k\in\Z}(k-\Psi)^2$ the resolvent kernel of the free operator $\left(\ii\frac{\mathrm{d}}{\mathrm{d}\varphi}+\Psi\right)^2$ with periodic boundary conditions is 
		\begin{align*}
			G_\lambda(\eta,\eta')=G_\lambda(\eta-\eta')=\frac{1}{2\pi}\sum_{n\in\Z}\frac{\e^{-\ii n(\eta-\eta')}}{(n+\Psi)^2-\lambda}\,.
		\end{align*}
		Using the Poisson summation formula we can compute that
		\begin{align*}
			G_0(0)=\frac{\pi}{2\sin(\Psi\pi)^2}\,.
		\end{align*}
		The number of eigenvalues of $h-v$ below $0$ is bounded from above by the trace of the Birman--Schwinger operator $K_0=\sqrt{v}h^{-1}\sqrt{v}$, i.e. by $G_0(0)\int_{-\pi}^\pi v(\varphi)\dd \varphi$.  
		
		In the remainder we may thus assume that $\int_{-\pi}^\pi v(\varphi)\dd \varphi\ge 1/G_0(0)=2\sin(\Psi\pi)^2/\pi$, otherwise there are no negative eigenvalues. We now iteratively decompose $[-\pi,\pi]$ into intervals $I_k=[x_k,x_{k+1}]$ with $x_1=-\pi$ and 
		\begin{align}
			|I_k|\int_{I_k} 2v(\varphi)\dd \varphi=\min(3,4\pi/G_0(0))=:\eps\,.
			\label{eq:Ik}
		\end{align}
		Here it is essential that we assumed $2\pi\int_{-\pi}^\pi v(\varphi)\dd \varphi\ge 2\pi/G_0(0)$ as otherwise it would not be possible to define $I_1$. 
		Since $|I_k|\ge\min(3,4\pi/G_0(0))/\int_{-\pi}^\pi 2v(\varphi)\dd \varphi$ it will eventually occur that at some $n$
		\begin{align*}
			(\pi-x_n)\int_{x_n}^\pi 2v(\varphi)\dd \varphi<\min(3,4\pi/G_0(0))\,.
		\end{align*}
		We then end the construction and define the last interval as $I_n=[y,\pi]$ with $y\le x_n$ such that \eqref{eq:Ik} holds.
		We thus have covered $[-\pi,\pi]$ by $n-1$ disjoint intervals $I_k$ and one additional interval $I_n$ that intersects with (only) $I_{n-1}$. The multiplicity of this covering is at most 2.  Importantly \eqref{eq:Ik} holds for all $I_k$.  For $u\in H^1(-\pi,\pi)$ we compute 
		\begin{align*}
		\sclp[2]{u,(h-v)u}=
			\int_{-\pi}^\pi(|\ii u'+\Psi u|^2-v|u|^2)\dd \varphi
			\ge \sum_{k=1}^n \int_{I_k}\left(\frac12|\ii u'+\Psi u|^2-v|u|^2\right)\dd \varphi. 
		\end{align*}
		Note that each of the quadratic forms in the sum above is closed on $H^1(I_k)$. 
		We now consider an operator $H:=\bigoplus_{k=1}^n h_k$ on $\Hs:=\bigoplus_{k=1}^n L^2(I_k)$. Each $h_k$ is defined as the self-adjoint operator corresponding to the $k$-th quadratic form in the above sum on $H^1(I_k)$. If $u\in H^1(-\pi,\pi)$ then $\bigoplus_{k=1}^nu|_{I_k}\in \bigoplus_{k=1}^n H^1(I_k)$ and
		\begin{align*}
			\frac12\norm[\Hs]{\bigoplus_{k=1}^nu|_{I_k}}^2\le\norm[2]{u}^2\le\norm[\Hs]{\bigoplus_{k=1}^nu|_{I_k}}^2\,.
		\end{align*}

		Using the variational principle and the above bounds, we can conclude that the negative eigenvalues $\lambda_m(h)$ of $h$ satisfy
		\begin{align*}
			\lambda_m(h-v)
			&\ge\inf_{\substack{F\subset H^1(-\pi,\pi)\\\dim F=m} }\sup_{u\in F\setminus\{0\}}\frac{-(\sclp[2]{u,hu})_-}{\norm[2]{u}^2}\\
			&\ge 2\inf_{\substack{F\subset H^1(-\pi,\pi)\\\dim F=m} }\sup_{u\in F\setminus\{0\}}\frac{-(\sclp[\Hs]{\bigoplus_{k=1}^nu|_{I_k},H\bigoplus_{k=1}^nu|_{I_k}})_-}{\norm[\Hs]{\bigoplus_{k=1}^nu|_{I_k}}^2}
		\end{align*}
		If $F$ is an $m$-dimensional subspace of $L^2(-\pi,\pi)$, then $\{\bigoplus_{k=1}^n u|_{I_k}:\,u\in F\}$ is an $m$-dimensional subspace of $\Hs$ and thus 
		\begin{align*}
			\lambda_m(h-v)\ge  2\inf_{\substack{F\subset \bigoplus_{k=1}^n H^1(I_k)\\\dim F=m} }\sup_{U\in F\setminus\{0\}}\frac{-(\sclp[\Hs]{U,HU})_-}{\norm[\Hs]{U}^2}=2\lambda_m(H)\,.
		\end{align*}
		The negative eigenvalues $\lambda_m(H)$ of $H$ coincide (including multiplicity) with the individual eigenvalues of the operators $h_k$. Via the transform $u\mapsto \e^{\ii\Psi\varphi}u$ each $h_k$ is seen to be unitarily equivalent to the operator
		\begin{align*}
			-\frac12\frac{\mathrm{d}^2}{\mathrm{d}\varphi^2}-v(\varphi)
		\end{align*}
		with Neumann boundary conditions on $I_k$. Thus, by \eqref{eq:Ik} and Lemma \ref{lem:W}, each $h_k$ has a single negative eigenvalue $\lambda_1(h_k)$ which is bounded by 
		\begin{align*}
			|2\lambda_1(h_k)|^{\frac12}\le\frac{g(|I_k|\int_{I_k}2v(\varphi)\dd\varphi)}{|I_k|}= \frac{g(\eps)}{|I_k|} =\frac{g(\eps)}{\eps}\int_{I_k} 2v(\varphi)\dd \varphi\,.
		\end{align*}
		We can conclude that
		\begin{align*}
			\sum_{m\ge 1}|\lambda_m(h-v)|^{\frac12}
			\le \sum_{k=1}^n|2\lambda_1(h_k)|^{\frac12}
			\le \sum_{k=1}^n \frac{g(\eps)}{\eps}\int_{I_k} 2v(\varphi)\dd \varphi 
			\le\frac{4g(\eps)}{\eps} \int_{-\pi}^\pi v(\varphi)\dd \varphi
		\end{align*}
		which finishes the proof. If $\Psi\to k\in\Z$ then $1/G_0(0)\to0$ and thus $\eps\to0$. Since $g(\eps)/\eps\to\infty$ as $\eps\to0$, we observe that $d_\Psi=4g(\eps)/\eps\to\infty$. 
	\end{proof}
\begin{remark}[Explicit constant]\label{Re:constant3}
    For particular $\Psi$, an upper bound on the constant in Theorem \ref{th:AB} can be found following the scheme of Remark \ref{Re:Constant}.
\end{remark}


	\section{Proof of Theorem \ref{th:as}}\label{sec:as}
	   Finally, consider the operator $H_{\mathrm{as}}-V=-\Delta-V$ on $L^2_{\mathrm{as}}(\mathbb{R}^2)$. It is well known that for antisymmetic functions $u$, in $H^1(\mathbb{R}^2)$, we have
		\begin{equation*}
			\int_{\mathbb{R}^2} |{\nabla u}|^2 \dd x\geq \int_{\mathbb{R}^2} \frac{|u|^2}{|x|^2}\dd x.
		\end{equation*}
		
		Repeating the arguments above, we arrive at an equation similar to (\ref{eq:thm2}) with $\Psi=0$, $c_\Psi=1$ and operator-valued potential
		\begin{equation*}
		    W(r)=\frac{1-\vartheta}{r^2}\frac{\partial^2}{\partial \varphi^2}+V(r,\varphi)
		\end{equation*}
		defined on the domain of periodic functions in $H^1(-\pi,\pi)$ which satisfy $u(\varphi)=-u(-\varphi)$.
		From Proposition \ref{prop:Cal} we have (with $f(t)=t$) 
	    \begin{equation*}
			N(H_{\mathrm{as}}-V)\le 2^{3/2}R\Big(\frac{\vartheta}{1-\vartheta}\Big)(1-\vartheta)^{-1/2}\int_{\mathbb{R}^+}\sqrt{\norm[\mathfrak{S}_{1/2}]{W_+(r)}}\dd r\,.
		\end{equation*}
	    Since functions in the domain of $W(r)$ vanish at $\pm\pi$ we can extend them onto the whole line by zero and use the standard Lieb--Thirring bound \eqref{hlt} (where $V$ is also extended by zero) to obtain
		\begin{equation*}
			N(H_{\mathrm{as}}-V)\le \sqrt{2}R\Big(\frac{\vartheta}{1-\vartheta}\Big)(1-\vartheta)^{-1}\int_{\mathbb{R}^+}\int_{\Sph^1} V(r,\varphi) r\dd \varphi \dd r\,.
		\end{equation*}
        \begin{remark}[Explicit constant]
            The constant in this case can be bounded above by $5.42152$, following Remark \ref{Re:Constant}.  
        \end{remark}
	
	\subsection*{Acknowlegements}A.~Laptev is grateful to M.~Z.~Solomyak for useful discussions of Calogero inequality for matrix-valued potentials.
	A.~Laptev was partially supported by RSF grant 18-11-0032. L.~Schimmer was supported by the VR grant 2017-04736 at Royal Swedish Academy of Sciences. The authors thank T.~Weidl for a helpful discussion. 


\begin{thebibliography}{30}
		
		%
		\bibitem[BEL]{BEL}
		A.~Balinsky, W.~D.~Evans, and R.~T.~Lewis, {\it On the number of negative eigenvalues of Schr\"odinger operators with an Aharonov-Bohm magnetic field}, Proc. Roy. Soc. London {\bf 457}, 2481--2489 (2001).
		%
		\bibitem[BL]{BL} M.~Birman and A.~Laptev, {\it The negative discrete spectrum
			of a two-dimensional Schr\"odinger operator,} Comm. Pure Appl.
		Math. {\bf 49}(9),  967--997 (1996).
		%
		\bibitem[C]{C} 
		F.~Calogero {\it Upper and lower limits for the number of bound states in a given central potential}, Comm. Math. Phys. {\bf 1}, 80--88 (1965).
		%
		\bibitem[CKMW]{CKMW} K.~Chadan, N.~N.~Khuri, A.~Martin and T.~T.~Wu,  {\it Bound states
			in one and two spatial dimensions,}  J. Math. Phys.  {\bf 44}(2), 406--422 (2003).
		%
		\bibitem[Cw]{Cw} 
		M.~Cwikel, {\it Weak type estimates for singular values and the number of bound states 
			of Schr\"odinger operators}, Ann. Math. (2), {\bf 106}(1), 93--100 (1977).
		%
		\bibitem[DELL]{DELL}
		J.~Dolbeault, M.~J.~Esteban, A.~Laptev and M.~Loss, 
		{\it Magnetic rings},
		J. Math. Phys. {\bf 59}, 051504 (2018).
		%
		\bibitem[FLW]{FLW}
		R.~L.~Frank, A.~Laptev and T.~Weidl, {\it Schr{\"o}dinger Operators: Eigenvalues and Lieb--Thirring inequalities},
		unpublished, (2021).
		%
		\bibitem[GN]{GN} A.~Grigor'yan and N.~Nadirashvili,
		{\it Negative eigenvalues of two-dimensional Schr\"odinger operators}, Arch. Ration. Mech. Anal.
		{\bf 217}(3), 975--1028 (2015).
		%
		\bibitem[HL]{HL} T.~Hoffmann-Ostenhof and A.~Laptev, 
		{\it Hardy Inequality for Antisymmetric Functions}, Funct. Anal. Appl.
		{\bf 55}, 122--129 (2021).
		%
		\bibitem[HLT]{HLT} 
		D.~Hundertmark, E.~H.~Lieb  and L.~E.~Thomas, {\it A sharp bound for an eigenvalue moment of the one-dimensional Schr{\"o}dinger operator}. Adv. Theor. Math. Phys. {\bf 2}, 719--731 (1998).
		%
		\bibitem[HLW]{HLW} 
		D.~Hundertmark, A.~Laptev and T.~Weidl, 
		\emph{New bounds on the Lieb-Thirring constants}, Invent. Math. {\bf 140}, 693--704, (2000).
		%
		\bibitem[KS]{KS} 
		M.~Karuhanga and E.~Shargorodsky, 
		{\it On negative eigenvalues of two-dimensional Schr\"odinger operators with singular potentials},
		J. Math. Phys. {\bf 61}(5), 051509 (2020).
		%
		\bibitem[KMW]{KMW}
		N.~N.~Khuri, A.~Martin and T.~T.~Wu, {\it Bound states in n dimensions (especially $n = 1$ and $n = 2$)}, Few Body Syst. {\bf 31}, 83--89 (2002).
		%
		\bibitem[L1]{L1} 
		A.~Laptev, {\it The negative spectrum of the class of
			two-dimensional Schr\"odinger operators with potentials that depend on the radius,} (Russian)
		Funktsional. Anal. i Prilozhen. {\bf 34}(4), 85--87  (2000);  translation in  Funct. Anal. Appl.
		{\bf 34}(4), 305--307 (2000).
		%
		\bibitem[L2]{L2} 
		A.~Laptev, {\it 
			Spectral inequalities for Partial Differential Equations and their applications}, Proceedings of ICCM2010 in Beijing, AMS/IP Studies in Advanced Mathematics  {\bf 51}, pt.2,  629--643 (2012).
		%
		\bibitem[LN]{LN}
		A.~Laptev and Yu.~Netrusov,  {\it On the negative eigenvalues of a class of Schr\"odinger operators}, In: Differential operators and spectral theory, Amer. Math. Soc. Transl. Ser. 2  {\bf 189}, Providence, RI: Amer. Math. Soc., 173--186 (1999).
		%
		\bibitem[LS1]{LS1} A.~Laptev and  M.~Solomyak, {\it On the negative spectrum of the two-dimensional Schr\"odinger operator with radial potential}, Comm. Math. Phys. {\bf 314}(1),  229--241 (2012). 
		
		\bibitem[LS2]{LS2} A.~Laptev and M.~Solomyak,  {\it
			On spectral estimates for two-dimensional Schr\"odinger operators},  J. Spectr. Theory {\bf 3}(4),  505--515 (2013).
		%
		\bibitem[LW1]{LW1} A.~Laptev and T.~Weidl,
		{\it Sharp Lieb-Thirring inequalities in high dimensions}, Acta Mathematica {\bf 184}, 87--111
		(2000). 
		%
		\bibitem[LW2]{LW2} A.~Laptev and T.~Weidl,
		{\it Hardy inequalities for magnetic Dirichlet forms}, Operator Theory: Adv. and Appl. {\bf 108}, 
		299--305 (1999).
		
		%
		\bibitem[L]{L}
		E.~H.~Lieb, {\it Bounds on the eigenvalues of the Laplace and Schr\"odinger operators}, Bull. Amer. Math. Soc. {\bf 82}(5), 751--753 (1976).
		%
		\bibitem[LT]{LT}
		E.~H.~Lieb and W.~E.~Thirring,
		{\it Inequalities for the moments of the eigenvalues of the Schr{\"o}dinger hamiltonian and their relation to Sobolev inequalities}, In: Studies in Mathematical Physics, pp. 269--303, Princeton University Press, Princeton (1976).
		%
		\bibitem[MW]{MW} 
		A.~Martin and T.~T.~Wu, {\it Bound states in two spatial
			dimensions in the noncentral case}, J. Math. Phys. {\bf 45}(3), 922--931
		(2004).
		%
		\bibitem[MV]{MV} 
		S.~Molchanov and B.~Vainberg, 
		{\it Bargmann type estimates of the counting function for general Schrödinger operators}, J. Math. Sci. {\bf 184}, 457--508 (2012).	
		%
		\bibitem[R]{R} 
		G.~V.~Rozenblum, {\it Distribution of the discrete spectrum of singular differential operators}, Dokl. Akad. Nauk SSSR {\bf 202}, 1012--1015 (1972). English translation in Soviet Math. Dokl. 
		{\bf 13}, 245--249 (1972).
		%
		\bibitem[Sol]{Sol} M.~Solomyak, {\it Piecewise-polynomial approximation of functions from
			$H^l((0,1)^d),\ 2l=d$, and applications to the
			spectral theory of the Schr\"odinger operator}, Israel J. Math.  {\bf 86}, (1-3), 253--275  (1994).
		%
		\bibitem[Sh]{Sh} 
		E.~Shargorodsky, 
		{\it On negative eigenvalues of two-dimensional Schr\"odinger operators}, 
		Proc. Lond. Math. Soc. (3) {\bf 108}(2), 441--483 (2014).
		%
		\bibitem[St]{St} 
		M.~Stoiciu, {\it An estimate for the number of bound
			states of the Schr\"odinger operator in two dimensions}, Proc. Amer. Math. Soc. {\bf 132}(4), 1143--1151 (2004).
		%
		\bibitem[W]{W}
		T.~Weidl, {\it On the {L}ieb-{T}hirring constants {$L_{\gamma,1}$} for {$\gamma\geq 1/2$}}, Comm. Math. Phys. {\bf 178}(1), 135--146, (1996).
		%
		\bibitem[We]{We}
		J.~Weidmann, {\it Spectral theory of ordinary differential operators}. Vol. 1258. Springer, 2006.
APA	

		
		
		
		
		
	\end{thebibliography}
\end{document}